\documentclass[a4paper,UKEnglish]{article}

\pdfoutput=1 

\usepackage{tabularx}

\usepackage{amsthm}
\usepackage{amsmath}
\usepackage{amssymb}
\usepackage{amsfonts}
\usepackage{mathtools}
\usepackage{enumitem}
\usepackage[pagebackref,linkcolor=blue,citecolor=red]{hyperref}
\usepackage{fullpage}

\usepackage{tikz}
\usepackage{multicol}

\usepackage[capitalize]{cleveref}

\usepackage[sort&compress,numbers]{natbib}

\usepackage{xspace}
\usepackage{comment}
\usepackage{todonotes}

\newtheorem{theorem}{Theorem}
\newtheorem{lemma}{Lemma}
\newtheorem{observation}{Observation}
\newtheorem{corollary}{Corollary}
\newtheorem{proposition}{Proposition}

\theoremstyle{definition}

\newcommand{\threefield}[3]{$#1\mid#2\mid#3$\xspace}
\newcommand{\Prob}{\threefield{1}{r_j}{\sum w_j U_j}}

\newcommand{\np}{\ensuremath{p_\#}}
\newcommand{\nw}{\ensuremath{w_\#}}
\newcommand{\nd}{\ensuremath{d_\#}}
\newcommand{\nr}{\ensuremath{r_\#}}

\newcommand{\optprob}[3]{
	\begin{center}
		\begin{minipage}{0.96\linewidth}
			\noindent
			\textsc{#1}
			\begin{description}
			 \item[\textbf{Input:}]  #2
			 \item[\textbf{Task:}]  #3
			\end{description}
		\end{minipage}
	\end{center}
}

\newlist{variables}{enumerate}{1}
\setlist[variables,1]{label=\text{\arabic*}, ref  =\text{\arabic*}}
\crefname{variablesi}{}{}

\newlist{constraints}{enumerate}{1}
\setlist[constraints,1]{label=\text{\arabic*}, ref  =\text{\arabic*}}
\crefname{constraintsi}{constraint}{constraints}

\usepackage{etoolbox}
\newcommand{\appsymb}{$\star$}
\newcommand{\appref}[1]{{\hyperref[proof:#1]{\appsymb}}}
\newcommand{\appLink}[1]{{\hyperref[#1]{\appsymb}}}

\title{Single-Machine Scheduling to Minimize the Number of Tardy Jobs with Release Dates} 

\author{Matthias Kaul\thanks{Universit{\"a}t Bonn. Most work was done at Hamburg University of Technology or at home. \texttt{mkaul@uni-bonn.de}}
  \and Matthias Mnich\thanks{Hamburg University of Technology, Institute for Algorithms and Complexity, Hamburg, Germany. \texttt{matthias.mnich@tuhh.de}. Partially supported by DFG project MN 59/4-1.}
  \and Hendrik Molter\thanks{Department of Computer Science, Ben-Gurion University of the Negev, Be'er-Sheva, Israel. \texttt{molterh@post.bgu.ac.il} Supported by the European Union's Horizon Europe research and innovation program under grant agreement 949707.}
}

\date{}

\begin{document}

\maketitle

\begin{abstract}
We study the fundamental scheduling problem \Prob: schedule a set of $n$ jobs with weights, processing times, release dates, and due dates on a single machine, such that each job starts after its release date and we maximize the weighted number of jobs that complete execution before their due date.
Problem \Prob generalizes both {\sc Knapsack} and {\sc Partition}, and the simplified setting without release dates was studied by Hermelin et al.\ [Annals of Operations Research, 2021] from a parameterized complexity viewpoint.

Our main contribution is a thorough complexity analysis of \Prob in terms of four key problem parameters: the number $\np$ of processing times, the number~$\nw$ of weights, the number $\nd$ of due dates, and the number $\nr$ of release dates of the jobs.
\Prob is known to be weakly para-$\mathsf{NP}$-hard even if~$\nw+\nd+\nr$ is constant, and Heeger and Hermelin [ESA, 2024] recently showed  (weak) $\mathsf{W}[1]$-hardness parameterized by~$\np$ or~$\nw$ even if $\nr$ is constant.

Algorithmically, we show that \Prob is fixed-parameter tractable parameterized by $\np$ combined with any two of the remaining three parameters $\nw$, $\nd$, and $\nr$.
We further provide pseudo-polynomial $\mathsf{XP}$-time algorithms for parameter $\nr$ and $\nd$.
To complement these algorithms, we show that \Prob is (strongly) $\mathsf{W}[1]$-hard when parameterized by $\nd+\nr$ even if $\nw$ is constant.
Our results provide a nearly complete picture of the complexity of \Prob for $\np$, $\nw$, $\nd$, and $\nr$ as parameters, and extend those of Hermelin et al.\ [Annals of Operations Research, 2021] for the problem $1\mid\mid\sum w_j U_j$ without release dates.
\end{abstract}

\section{Introduction}
\label{sec:introduction}
The problem of scheduling jobs to machines is one of the core application areas of combinatorial optimization~\cite{pinedo2012scheduling}. 
Typically, the task is to allocate jobs to machines in order to maximize a certain objective function while complying with certain constraints.
In our setting, the jobs are characterized by several numerical parameters: a \emph{processing time}, a \emph{release date}, a \emph{due date}, and a \emph{weight}.
We have access to a single machine that can process one job (non-preemptively) at a time.
We consider one of the most fundamental objective functions, namely to minimize the weighted number of tardy jobs, where a job is considered \emph{tardy} if it completes after its due date.
In the standard three-field notation for scheduling problems by Graham~\cite{Graham1969}, the problem is called \Prob.
We give a formal definition in \cref{sec:preliminaries}.

The interest in \Prob comes from various sources.
It generalizes several fundamental combinatorial problems.
In the most simple setting, without weights and release times, the classic algorithm by Moore~\cite{Moore68} computes an optimal schedule in polynomial time.
When weights are added, the problem ($1\mid\mid\sum w_jU_j$) encapsulates the {\sc Knapsack} problem, a cornerstone in combinatorial optimization and one of Karp's 21 $\mathsf{NP}$-complete problems~\cite{Karp72}.
Precisely, when all jobs are released at time zero and all jobs have a common due date, we obtain the {\sc Knapsack} problem.
Karp’s $\mathsf{NP}$-hardness proof (from his seminal paper~\cite{Karp72}) is the first example of a reduction to a problem involving numbers.
The problem $1\mid\mid\sum w_jU_j$ is one of the most extensively studied problems in scheduling and can be solved in pseudopolynomial time by the classic algorithm by Lawler and Moore~\cite{LawlerMoore}.
Hermelin et al.~\cite{hermelin2023minimizing} showed that this algorithm can be improved in various restricted settings.
Better running times have been achieved for the special case where the weights of the jobs equal their processing times~\cite{bringmann2022faster,FischerW2024,klein2023minimizing,schieber2023quick}.

The problem $1\mid\mid\sum w_jU_j$ (so without release times) has been studied from the perspective of parameterized complexity by Hermelin et al.~\cite{HermelinKPS21}.
They considered the number $\np$ of different processing times, the number $\nd$ of different due dates, and the number $\nw$ of different weights as parameters and showed fixed-parameter tractability for $\nw+\np$, $\nw+\nd$, and $\np+\nd$ as well as giving an $\mathsf{XP}$-algorithm for the parameters $\np$ and $\nw$.
These results are presumably tight, as Heeger and Hermelin~\cite{HeegerH24} recently showed (weak) $\mathsf{W}[1]$-hardness for the parameters~$\np$ and $\nw$.
The problem has also been studied under fairness aspects~\cite{heeger2023equitable}.

The addition of release dates is naturally motivated in every scenario where not all jobs are initially available.
The aim of this paper is to study the parameterized complexity of the problem (\Prob) in this setting.
Here, it encapsulates {\sc Partition} and becomes weakly $\mathsf{NP}$-hard~\cite{LenstraRinnooy-Kan77}, even if there are only two different release dates and two different due dates and all jobs have the same weight.
It has previously been studied for the case of uniform processing times, both on a single machine~\cite{GareyUnitJobs81} and for parallel machines~\cite{baptiste2004ten,HM24}, as well as for the special case of interval scheduling~\cite{arkin1987scheduling,hermelin2024parameterized,krumke2011interval,sung2005maximizing}.

\paragraph{Our Contributions.}
In this paper, we deploy the tools of parameterized complexity to study the computational complexity of \Prob.
In the spirit of ``parameterizing by the number of numbers''~\cite{Fellows12}, we analyze the complexity picture with respect to (i) the number~$\np$ of distinct processing times of the jobs, (ii) the number~$\nw$ of distinct weights of the jobs, (iii) the number $\nd$ of distinct due dates of the jobs, and (iv) the number $\nr$ of distinct release dates of the jobs.
Thereby, we extend and complement the results obtained by Hermelin et al.~\cite{HermelinKPS21} for the case where all release dates are zero.

In summary, we obtain an almost complete classification into tractable cases (meaning that we find a fixed-parameter algorithm) and intractable cases (meaning that we show the problem to be $\mathsf{W}[1]$-hard) depending on which subset of parameters from $\{\np, \nw, \nd, \nr\}$ we consider.

First note that for some parameter combinations, the problem complexity has already been resolved.
In particular, \Prob is known to be weakly $\mathsf{NP}$-hard $\nd=1$ and $\nr=1$, as this setting captures the {\sc Knapsack} problem~\cite{Karp72}.
Furthermore, \Prob is known to be weakly $\mathsf{NP}$-hard for $\nd=2$, $\nw=1$, and $\nr=2$~\cite{LenstraRinnooy-Kan77}.
For parameter $\np$ as well as $\nw$, Heeger and Hermelin~\cite{HeegerH24} showed weak $\mathsf{W}[1]$-hardness even if $\nr=1$.

That leaves open the parameterized complexity for several parameter combinations.
We extend the known hardness results by showing the following:
\begin{itemize}
  \item \Prob is (strongly) $\mathsf{W}[1]$-hard parameterized by $\nd+\nr$ even if $\nw=1$. 
\end{itemize}
This result is obtained by a straightforward reduction showing that \Prob generalizes {\sc Bin Packing}.
Our main results are on the algorithmic side, where we show the following:
\begin{itemize}
  \item \Prob is fixed-parameter tractable parameterized by  $\np+\nd+\nr$.
  \item \Prob is fixed-parameter tractable parameterized by $\np+ \nw+ \nd$ or $\np+ \nw+ \nr$.
  \item \Prob can be solved in pseudo-polynomial time for constant $\nr$ or constant $\nd$.
\end{itemize}
For the first two, we employ reductions to {\sc Mixed Integer Linear Programming (MILP)}, and for the latter, we give a dynamic programming algorithm.

With our results, we resolve the parameterized complexity of \Prob for almost all parameter combinations of $\{\np,\nw,\nr,\nd\}$ and hence give the first comprehensive overview thereof.
The remaining question is whether \Prob is polynomial-time solvable for \emph{constant} $\np$ or fixed-parameter tractable for $\np+\nw$.
It also remains open whether \Prob is fixed-parameter tractable parameterized by $\np$ if all numbers are encoded in unary.
A technical report~\cite{ElffersW14} claims (strong) $\mathsf{NP}$-hardness even for $\np=2$ and $\nw = 1$.
If that result holds, then it would also settle the parameterized complexity for the parameter combination~$\np+\nw$, and for $\np$ when all processing times and weights are encoded in unary.
Otherwise, those questions would remain open.

Our results contribute to the growing body of investigating the parameterized complexity of fundamental scheduling problems~\cite{MnichW2015}; for reference, we refer to the open problem collection by Mnich and van Bevern~\cite{MnichvB2018}.

\section{Preliminaries}
\label{sec:preliminaries}

\paragraph{Scheduling.} The problem considered in this work is denoted \Prob in the standard three-field notation for scheduling problems by Graham~\cite{Graham1969}.
In this problem, we have $n$ jobs and one machine that can process one job at a time. 
Each job $j\in\{1,\ldots,n\}$ has a \emph{processing time} $p_j$, a \emph{release date}~$r_j$, a \emph{due date} $d_j$, and a \emph{weight} $w_j$, where we $p_j$, $r_j$, $d_j$, and~$w_j$ are non-negative integers.
We use \mbox{$\np$, $\nr$, $\nd$,} and~$\nw$ to denote the number of different processing times, release dates, due dates, and weights, respectively. 

A \emph{schedule} $\sigma : \{1,\ldots,n\}\rightarrow 
\mathbb{N}$ assigns to each job $j$ a \emph{starting time} $\sigma(j)$ to process it until its \emph{completion time} $\sigma(j)+p_j$, so no other job $j'\not=j$ must start during $j$'s \emph{execution time} $\sigma(j),\hdots,\sigma(j)+p_j-1$.
We call a job $j$ \emph{early} in a schedule $\sigma$ if $\sigma(j)+p_j\le d_j$; otherwise we call job $j$ \emph{tardy}. 
We say that the machine is \emph{idle} at time $s$ in a schedule $\sigma$ if no job's execution time contains $s$.
The goal is to find a schedule that minimizes the weighted number of tardy jobs or, equivalently, maximizes the weighted number of early jobs 
\begin{equation*}
  W=\sum_{j\mid \sigma(j)=s \wedge s+p_j\le d_j}w_j \enspace .
\end{equation*}
We call a schedule that maximizes the weighted number of early jobs \emph{optimal}.
Formally, the problem is defined as follows:

\smallskip

\optprob{\Prob}{A number $n$ of jobs, a list of processing times $(p_1,p_2,\ldots,p_n)$, a list of release dates $(r_1,r_2,\ldots,r_n)$, a list of due dates $(d_1,d_2,\ldots,d_n)$, and a list of weights~$(w_1,w_2,\ldots,w_n)$.}{Compute an optimal schedule, that is, a schedule $\sigma$ that maximizes $W=\sum_{j\mid \sigma(j)=s \wedge s+p_j\le d_j}w_j$.}

Given an instance $I$ of \Prob, we make the following observation: 
\begin{observation}
\label{obs:releaseduedate}
  Let $I$ be an instance of \Prob and let $d_{\max}$ be the largest due date of any job in~$I$.
  Let $I'$ be the instance obtained from~$I$ by setting $r'_j=d_{\max}-d_j$ and $d'_j=d_{\max}-r_j$ for each job~$j$.
  Then $I$ admits a schedule where the weighted number of early jobs is~$W$ if and only if $I'$ admits a schedule where the weighted number of early jobs is $W$.
\end{observation}
\cref{obs:releaseduedate} holds, as we can transform a schedule $\sigma$ for $I$ into a schedule~$\sigma'$ for $I'$ (with the same weighted number of early jobs) by setting $\sigma'(j)=d_{\max}-\sigma(j)-p_j$.
Intuitively, this means that we can switch the roles of release dates and due dates to obtain instances with the same objective value.

\paragraph{Mixed Integer Linear Programming.}
For several of our algorithmic results, we use reductions to {\sc Mixed Integer Linear Programming (MILP)}.
This problem is defined as follows:

\optprob{\textsc{Mixed Integer Linear Programming} (\textsc{MILP})}{A vector $x$ of $n$ variables, a subset~$S$ of the variables which are considered integer variables, a constraint matrix $A\in\mathbb{R}^{m\times n}$, and two vectors $b\in\mathbb{R}^m$, $c\in \mathbb{R}^n$.}{Compute an assignment to the variables (if one exists) such that all integer variables in $S$ are set to integer values, $Ax\le b$, $x\ge 0$, and $c^\intercal x$ is maximized.}

If all variables are integer variables, the problem is simply called {\sc Integer Linear Programming (ILP)}.
Due to Lenstra's well-known result for MILP~\cite{lenstra1983integer}, we have that:
\begin{theorem}[\cite{lenstra1983integer}]
\label{thm:micp}
  MILP is fixed-parameter tractable when parameterized by the number of integer variables.
\end{theorem}

\clearpage
\pagebreak

\section{Hardness Results}
\label{sec:hardness}
In this section, we first discuss known hardness results for \Prob, and then present a novel parameterized hardness result.
Observe that for $r_j=0$ and $d_j=d$, the problem \Prob is equivalent to \textsc{Knapsack}, which is known to be weakly $\mathsf{NP}$-hard~\cite{Karp72}.
Further, there is a straightforward reduction from {\sc Partition} to \Prob that only uses two release dates and two due dates, and uniform weights~\cite{LenstraRinnooy-Kan77}. 
Finally, Heeger and Hermelin~\cite{HeegerH24} recently showed that the special case of \Prob without release dates is weakly $\mathsf{W}[1]$-hard when parameterized by either $\np$ or $\nw$.
Hence, (together with \cref{obs:releaseduedate}) we have that:
\begin{proposition}[\cite{HeegerH24,Karp72,LenstraRinnooy-Kan77}]
  The problem \Prob is 
  \begin{itemize}
    \item weakly $\mathsf{NP}$-hard even if $\nd=1$ and $\nr=1$,
    \item weakly $\mathsf{NP}$-hard even if $\nr=\nd=2$ and $\nw=1$,
    \item and weakly $\mathsf{W}[1]$-hard when parameterized by either $\np$ or $\nw$ even if $\nr=1$ or if $\nd=1$.
  \end{itemize}    
\end{proposition}

The reduction from {\sc Partition} to \Prob by Lenstra et al.~\cite{LenstraRinnooy-Kan77} can straightforwardly be extended to a reduction from {\sc Bin Packing}, which yields the following result:

\begin{theorem}
\label{thm:hardness}
  The problem \Prob\ is strongly $\mathsf{NP}$-hard, and strongly $\mathsf{W}[1]$-hard when parameterized by $\nr+\nd$, even if $\nw=1$.
\end{theorem}
\begin{proof}
  We present a parameterized reduction from {\sc Bin Packing} parameterized by the number of bins.
  Here, we are given a set $I=\{1,\ldots,|I|\}$ of items with unary encoded sizes~$s_j$ for $j\in I$, and $b$ bins of size~$B$.
  We are asked whether it is possible to distribute the items to the bins such that no bin is overfull, that is, the sum of the sizes of items put into the same bin does not exceed~$B$.
  The reduction uses a similar idea as the one from {\sc Partition} used by Lenstra et al.~\cite{LenstraRinnooy-Kan77}.
  Intuitively, we use special jobs that can only be scheduled at very specific points in time to separate the total available time into several ``sections'', representing the bins.

  Formally, given an instance of {\sc Bin Packing}, we create an instance of \Prob as follows.
  For every item in $j\in I$ we create a job $j$ with processing time $p_j=s_j$, weight one, release date zero, and due date $d_j=b\cdot B+b-1$.
  Furthermore, we create $b-1$ ``separator'' jobs $j^\star_1, \ldots, j^\star_{b-1}$.
  Job $j^\star_i$ has processing time one, weight one, release date $r_{j^\star_i}=i\cdot B+i-1$, and due date $d_{j^\star_i}=i\cdot B+i$.
  Note that all jobs have weight one.
  Furthermore, all jobs corresponding to items of the {\sc Bin Packing} instance have the same release date and due date and there are $b-1$ additional separator jobs, each with a different release date and due date.
  Hence, we have that $\nd+\nr\in O(b)$.
  The \Prob instance can clearly be constructed in polynomial time.
  In the remainder of the proof, we show that the {\sc Bin Packing} instance is a ``yes''-instance if and only if there is a schedule for the \Prob instance where no job is tardy.
  As {\sc Bin Packing} strongly $\mathsf{NP}$-hard~\cite{GJ79}, and strongly $\mathsf{W}[1]$-hard when parameterized by the number of bins~\cite{jansen2013bin}, the result follows.

  In the forward direction, suppose that the {\sc Bin Packing} instance is a ``yes''-instance and we have a partition of the items into the bins such that no bin is overfull.
  We schedule the jobs as follows. We schedule separator job $j^\star_i$ to starting time $i\cdot B+i-1$.
  All jobs corresponding to items placed into the $i^{\textnormal{th}}$ bin are scheduled into the time interval from $(i-1)\cdot B+i-1$ to $i\cdot B+i-1$ in an arbitrary order.
  Note that since the bins have size $B$ and the described time interval has length $B$, we can schedule all jobs into this interval.
  Furthermore, since all jobs corresponding to items have release date zero and due date $b\cdot B+b-1$, we have that all those jobs respect their release dates and due dates.
  Lastly, we have that the described intervals do not contain the time slots in which the separator jobs are scheduled.
  It follows that we have a schedule where no job is tardy.

  In the backward direction, suppose that there is a schedule for the \Prob instance where no job is tardy.
  Then the separator jobs $j^\star_i$ are scheduled to starting time $i\cdot B+i-1$, since by construction this is the only way to schedule the separator jobs such that they are not tardy.
  We create a distribution of the items into the bins as follows.
  We assign all items corresponding to jobs scheduled in the time interval from $(i-1)\cdot B+i-1$ to $i\cdot B+i-1$ to bin $i$.
  Since the time interval has size $B$ and no job can intersect with the separator jobs, we know that no bin is overfull.
  Furthermore, since all jobs are scheduled we have that each item is in one bin.
  It follows that the {\sc Bin Packing} instance is a ``yes''-instance.
\end{proof}

\section{{\boldmath \Prob} parameterized by {\boldmath $\np+\nr+\nd$}}
\label{sec:fpt1}
In this section, we present the following result.
\begin{theorem}
\label{thm:fpt1}
  The problem \Prob is fixed-parameter tractable when parameterized by \mbox{$\np+\nr+\nd$}.
\end{theorem}

To prove \cref{thm:fpt1}, we present a reduction from \Prob to MILP that creates instances of MILP where the number of integer variables is upper-bounded by a function of $\np$, $\nr$, and $\nd$.
The result then follows from \cref{thm:micp}.

Given an instance of \Prob, we say that two jobs $j$ and $j'$ have the same type if they have the same processing time, the same release date, and the same due date, that is, $p_j=p_{j'}$, $r_j=r_{j'}$, and $d_j=d_{j'}$.
Let $\mathcal T$ denote the set of all types.
Note that we have $|\mathcal{T}|\le \np\cdot\nr\cdot\nd$. 
Furthermore, we sort all release dates and due dates such that if the $k^{\textnormal{th}}$ release date equals the $\ell^{\textnormal{th}}$ due date we require the release date to appear later in the ordering.
Let $\mathcal L$ denote an ordered list of all release dates and due dates that complies to the aforementioned requirement.
Note that we have $|\mathcal{L}|\le \nr+\nd$.

We now create an integer variable $x^t_{a,b}$ for all $t\in\mathcal{T}$ and all $a,b\in\mathcal{L}$ with $a<b$.
Intuitively, if $a$ and $b$ are consecutive in $\mathcal{L}$, then this variable tells us how many jobs of type $t$ to schedule in the time interval $[a,b]$.
If $a$ and $b$ are not consecutive in $\mathcal{L}$, then $x^t_{a,b}$ is a zero-one variable that tells us whether we schedule a job of type $t$ in a way such that its processing time intersects all $c\in\mathcal{L}$ with $a<c<b$. 

We create the following constraints.
The first set of constraints is 
\begin{equation}
\label{eq:1}
  \forall t\in \mathcal{T}: \sum_{a,b\in\mathcal{L}\mid a< b} x^t_{a,b}\le n_t,
\end{equation}
where $n_t$ denotes the number of jobs of type $t$ in the \Prob instance.
Intuitively, these constraints ensure that we do not try to schedule more jobs of type $t$ than there are available.

The second set of constraints is 
\begin{equation}
\label{eq:2}
  \forall a,b\in \mathcal{L} \text{ with } a<b \text{ and } \forall t\in\mathcal{T} \text{ with } r_t>b \text{ or } d_t<a: x^t_{a,b}= 0,
\end{equation}
where $r_t$ denotes the release date of jobs of type $t$ and $d_t$ denotes the due date of jobs of type~$t$.
Intuitively, these constraints prevent us from trying to schedule a job of a certain type into an interval that conflicts with the job's release date or due date.

The third set of constraints is
\begin{equation}
\label{eq:3}
  \forall c\in\mathcal{L}: \sum_{t\in\mathcal{T}}\sum_{a,b\in\mathcal{L}\mid a<c<b} x^t_{a,b}\le 1 \enspace .
\end{equation}
Intuitively, these constraints ensure that for each $c\in\mathcal{L}$ at most one job is scheduled that intersects $c$.

The fourth set of constraints is 
\begin{align}
\label{eq:4}
  & \forall a,b\in \mathcal{L} \text{ with } a<b: \nonumber \\ & \sum_{t\in \mathcal{T}}\left(\sum_{a',b'\in\mathcal{L}\mid a'\ge a,b'\le b,a'<b'} p_t\cdot x^t_{a',b'} +
      \sum_{a',b'\in\mathcal{L}\mid a'<a,b'>b}(b-a)\cdot x^t_{a',b'}\right)\le b-a,
\end{align}
where $p_t$ denotes the processing time of jobs of type $t$.
Intuitively, these constraints make sure that for every interval $[a,b]$ with $a,b\in\mathcal{L}$ we do not schedule jobs with total processing time more than $b-a$ into that interval.

Now we specify the objective function.
If we want to schedule a certain number of jobs with type $t$ early, we take the ones with the largest weight in order to minimize the weighted number of tardy jobs.
Let $x^t=\sum_{a,b\in\mathcal{L}; a< b} x^t_{a,b}$.
Intuitively, $x^t$ is the number of jobs with type $t$ that are scheduled early.
For each type, $t\in\mathcal{T}$, order the jobs of type $t$ in the \Prob instance by their weight (largest to smallest) and let $w^t_i$ denote the $i$th largest weight of jobs with type $t$.
The objective function we aim to maximize is
\begin{equation}
\label{eq:obj}
\sum_{t\in\mathcal{T}}\sum_{i=1}^{x^t} w_i^t \enspace .
\end{equation}
    
Note that constraints \eqref{eq:1}, \eqref{eq:2}, \eqref{eq:3}, and \eqref{eq:4} are linear. 
The objective function \eqref{eq:obj} is convex~\cite{HermelinKPS21} and it is known that we can obtain an equivalent MILP with a linear objective function at the cost of introducing additional fractional variables and constraints~\cite{HermelinKPS21}. 
Furthermore, we can observe the following:
\begin{observation}
\label{obs:vars}
  The number of integer variables in the created MILP instance is in $O(\np\cdot\nr\cdot\nd\cdot(\nr+\nd)^2)$.
\end{observation}

Next, we show the correctness of the reduction.
\begin{lemma}
\label{lem:corr1}
  If the \Prob instance admits a schedule where the weighted number of early jobs is~$W$, then the created MILP instance admits a feasible solution that has objective value at least $W$.
\end{lemma}
\begin{proof}
  Suppose we are given a schedule $\sigma$ for the \Prob instance where the weighted number of early jobs is $W$.
  We create a feasible solution for the MILP instance as follows.

  Initially, we set all variables to zero.
  Now we iterate through the jobs that are early in the schedule~$\sigma$ in some arbitrary way.
  Let job $j$ be early in $\sigma$ and let $t\in \mathcal{T}$ by the type of job $j$.
  Let $\sigma(j)=s$ and let $a\in\mathcal{L}$ be the largest element of $\mathcal{L}$ in the ordering such that $a\le s$.
  Let $b\in\mathcal{L}$ be the smallest element of~$\mathcal{L}$ in the ordering such that $b\ge s+p_j$.
  We increase the variable $x_{a,b}^t$ by one.

  We first show that the produced solution for the MILP instance is feasible.
  Since for every job $j$ of type $t\in \mathcal{T}$, we increase at most one variable $x^t_{a,b}$ for some $a,b\in\mathcal{L}$ with $r_j\le a$ and $d_j\ge b$ by one, the constraints~\eqref{eq:1} and~\eqref{eq:2} are clearly met.
  Now assume, for the sake of contradiction, that for some $c\in\mathcal{L}$, the corresponding constraint \eqref{eq:3} is not met.
  Then there are $a,a',b,b'\in\mathcal{L}$ with $a<c<b$ and $a'<c<b'$ and $t,t'\in\mathcal{T}$ such that both $x_{a,b}^t\ge 1$ and $x_{a',b'}^{t'}\ge 1$.
  Let $j$ and $j'$ be the jobs that lead to the increase of $x_{a,b}^t$ and $x_{a',b'}^{t'}$, respectively.
  Then we have $\sigma(j)<c<\sigma(j)+p_j$ and $\sigma(j')<c<\sigma(j')+p_{j'}$, that is, jobs $j$ and $j'$ are in conflict.
  This contradicts the assumption that $\sigma$ is a schedule.
  Lastly, assume that for some $a,b\in\mathcal{L}$, the corresponding constraint \eqref{eq:4} is not met.
  Note that since the constraints \eqref{eq:3} are met, there is at most one combination of $a',b'\in\mathcal{L}$ and $t\in\mathcal{T}$ with $a'<a$ and $b'>b$ such that $x_{a',b'}^t\ge 1$ and in particular, if $x_{a',b'}^t\ge 1$ then $x_{a',b'}^t= 1$.
  Hence, we make the following case distinction.
  \begin{enumerate}
    \item There is a combination of $a',b'\in\mathcal{L}$ and $t\in\mathcal{T}$ with $a'<a$ and $b'>b$ such that $x_{a',b'}^t= 1$.

      Let $j$ denote the job of type $t$ that caused us to increase $x_{a',b'}^t$ to one.
      In particular, this means that $\sigma(j)<a$ and $\sigma(j)+p_j>b$.
      If the constraint \eqref{eq:4} is not met, then there is a combination of $a'',b''\in\mathcal{L}$ and $t'\in\mathcal{T}$ with $a''\ge a$, $b''\le b$, and $a''<b''$ such that $x_{a'',b''}^{t'}\ge 1$.
      Let $j'$ denote the job of type $t'$ that caused us to increase $x_{a'',b''}^{t'}$ by one. 
      Then we have that $\sigma(j')\ge a$ and $\sigma(j')+p_{j'}\le b$.
      Clearly, jobs $j$ and $j'$ are in conflict, contradicting the assumption that $\sigma$ is a schedule.

    \item For all combinations of $a',b'\in\mathcal{L}$ and $t\in\mathcal{T}$ with $a'<a$ and $b'>b$ we have that $x_{a',b'}^t= 0$.

      In this case, we must have that
      \begin{equation*}
        \sum_{t\in \mathcal{T}}\sum_{a',b'\in\mathcal{L}\mid a'\ge a,b'\le b,a'<b'} p_t\cdot x^t_{a',b'} > b-a \enspace .
      \end{equation*}
      This implies that the total processing times of jobs scheduled to start at time $a$ or later and finish at time $b$ or earlier exceeds $b-a$.
      By the pigeonhole principle, there have to be two jobs that are in conflict, contradicting the assumption that $\sigma$ is a schedule.
  \end{enumerate}
  We can conclude that the constructed solution for the MILP instance is feasible.
  
  It remains to show that the solution has objective value at least $W$.
  To this end, consider $x^t=\sum_{a,b\in\mathcal{L}; a< b} x^t_{a,b}$.
  By construction of the solution, we have that $x^t$ is the number of jobs of type $t$ that are scheduled early in $\sigma$.
  The weight of those jobs cannot be larger than $\sum_{i=1}^{x^t} w_i^t$, since this is precisely the weight of the $x^t$ heaviest jobs of type $t$ in the \Prob instance.
  It follows that 
  \begin{equation*}
    \sum_{t\in\mathcal{T}}\sum_{i=1}^{x^t} w_i^t\ge W \enspace .
  \end{equation*}
  This concludes the proof.
\end{proof}

\begin{lemma}
\label{lem:corr2}
  If the created MILP instance is feasible and admits a solution with objective value~$W$, then the \Prob instance admits a schedule where the weighted number of early jobs is at least $W$.
\end{lemma}
\begin{proof}
  Suppose we are given a solution to the MILP instance with objective value $W$.
  We create a schedule $\sigma$ as follows.

  We iterate through pairs $a,b\in\mathcal{L}$ with $a<b$ as follows.
  We start with the smallest element $a\in\mathcal{L}$ and the smallest element $b\in\mathcal{L}$ with $a<b$ according to the ordering.
  We maintain a current starting point $s$ that is initially set to $s=a$.
  Furthermore, we consider all jobs initially as ``unscheduled''.
  We proceed as follows.
  \begin{itemize}
    \item We iterate through all $t\in\mathcal{T}$ in some arbitrary but fixed way.
      If $x^t_{a,b}>0$, take the $x^t_{a,b}$ unscheduled jobs of type $t$ with the maximum weight and schedule those between $s$ and $s+x^t_{a,b}\cdot p_t$, where $p_t$ is the processing time of jobs of type $t$.
      Now consider those jobs scheduled and set $s\leftarrow s+x^t_{a,b}\cdot p_t$.
      Continue with the next type.
    \item Replace $b$ with the next-larger element in $\mathcal{L}$.
      If $b$ is the largest element in $\mathcal{L}$, then replace $a$ with the next-larger element in $\mathcal{L}$, set $b$ to the smallest element $b\in\mathcal{L}$ with $a<b$ according to the ordering, and set $s\leftarrow \max\{a,s\}$.
      If $a$ is the largest element of $\mathcal{L}$, terminate the process.
      Otherwise, go to the first step.
  \end{itemize}
  Since the solution to the MILP obeys constraint \eqref{eq:1}, we have that there are sufficiently many jobs of each type that can be scheduled.
  All jobs that remained unscheduled after the above-described procedure are scheduled in some arbitrary but feasible way. 
        
  We claim that the above procedure produces a schedule where the weighted number of early jobs is at least $W$.
  We start by showing that the procedure indeed produces a schedule.
  Clearly, there are no two jobs in conflict in the produced schedule.
  Furthermore, since we always have $s\ge a$ and since the solution to the MILP obeys constraints \eqref{eq:2}, we have that no job is scheduled to start before their release date.

  In the remainder of the proof, we show that we schedule $x^t=\sum_{a,b\in\mathcal{L}; a< b} x^t_{a,b}$ jobs of type $t$ early.
  In particular, we schedule the $x^t$ jobs of type $t$ with the largest weights early.
  As the solution to the MILP has objective value $W$, it follows that the total weight of early jobs is at least $W$.

  We show that all jobs scheduled in the first step of the above-described procedure are early.
  To this end, we identify where the procedure inserts idle times and argue that all jobs scheduled in the first step between two consecutive idle times by the procedure are early.
  An idle time is inserted by the procedure whenever we set $s\leftarrow \max\{a,s\}$ and we have $a>s$.
  Consider the case where we set $s\leftarrow \max\{a,s\}=a$ for some $a\in\mathcal{L}$.
  (Note that, for technical reasons, this includes the case where $a=s$.)
  Let $a'\in\mathcal{L}$ denote the next larger element of~$\mathcal{L}$ for which we set $s\leftarrow \max\{a',s\}=a'$.
  Now suppose, for the sake of contradiction, that there is a job with a starting time between $a$ and $a'$ that is scheduled in the first step of the procedure and is tardy.
  Let $j$ be the first such job, that is, the one with the smallest starting time. Let $x_{a'',b}^t\ge 1$ with some $a\le a''\le a'$ be the variable that was considered by the procedure when~$j$ was scheduled.
  Then we have that $d_j\ge b$, since otherwise constraints~\eqref{eq:2} are not met.
  We make the following case distinction:
  \begin{enumerate}
    \item If $a=a''$, then the completion time of $j$ is at most \mbox{$a+\sum_{t\in\mathcal{T}}\sum_{b'\in\mathcal{L}\mid b'\le b,a<b'}p_t \cdot x_{a,b'}^t$}. 
      However, since constraints~\eqref{eq:4} are met by the solution to the MILP, in particular the one for $a,b\in\mathcal{L}$, we have that
      \begin{equation*}
        a+\sum_{t\in\mathcal{T}}\sum_{b'\in\mathcal{L}\mid b'\le b,a<b'}p_t \cdot x_{a,b'}^t\le b \enspace .
      \end{equation*}
      Since $b\le d_j$, this contradicts the assumption that $j$ is tardy.
    \item Assume that $a<a''\le a'$.
      We argue that in this case, for all $x_{a''',b'}^{t'}$ with $t'\in\mathcal{T}$ and $a''',b'\in\mathcal{L}$ with $a\le a'''<a''$ and $b'>b$ we must have that $x_{a''',b'}^{t'}=0$.
      Assume that $x_{a''',b'}^{t'}$ for some $t'\in\mathcal{T}$ and $a''',b'\in\mathcal{L}$ with $a\le a'''<a''$ and $b'>b$.
      Consider the constraint~\eqref{eq:4} for $a'',b\in\mathcal{L}$.
      Then we have
      \begin{equation*}
        \sum_{t\in \mathcal{T}}\left(\sum_{a',b'\in\mathcal{L}\mid a'\ge a'',b'\le b,a'<b'} p_t\cdot x^t_{a',b'} + \sum_{a',b'\in\mathcal{L}\mid a'<a'',b'>b}(b-a'')\cdot x^t_{a',b'}\right) \le b-a'' \enspace .
      \end{equation*}
      However, we also have that
      \begin{align*}
        b-a'' & < p_t\cdot x_{a'',b}^t+(b-a'')\cdot x_{a''',b'}^{t'}\\
              & \le \sum_{t\in \mathcal{T}}\left(\sum_{a',b'\in\mathcal{L}\mid a'\ge a'',b'\le b,a'<b'} p_t\cdot x^t_{a',b'} + \sum_{a',b'\in\mathcal{L}\mid a'<a'',b'>b}(b-a'')\cdot x^t_{a',b'}\right),
      \end{align*}
      contradicting the assumption that $x_{a''',b'}^{t'}\ge 1$.
      It follows that the completion time of job~$j$ is at most $a+ \sum_{t\in \mathcal{T}}\sum_{a',b'\in\mathcal{L}\mid a'\ge a,b'\le b,a'<b'} p_t\cdot x^t_{a',b'}$.
      Since constraint~\eqref{eq:4} is met for $a,b\in\mathcal{L}$, we have that
      \begin{equation*}
        a+ \sum_{t\in \mathcal{T}}\sum_{a',b'\in\mathcal{L}\mid a'\ge a,b'\le b,a'<b'} p_t\cdot x^t_{a',b'}\le b \enspace .
      \end{equation*}
      Since $b\le d_j$, this contradicts the assumption that $j$ is tardy.
  \end{enumerate}
  We conclude that all jobs scheduled by the above-described procedure in the first step are early.
  Since for every $a,b\in\mathcal{L}$ and $t\in \mathcal{T}$, the procedure schedules the $x_{a,b}^t$ jobs of type $t$ with the largest weights early, we have that the total weight of early jobs is at least $W$.
  This concludes the proof.
\end{proof}

Now we have all the pieces available that are needed to prove \cref{thm:fpt1}.

\begin{proof}[Proof of \cref{thm:fpt1}]
  \cref{thm:fpt1} follows directly from \cref{obs:vars}, \cref{lem:corr1,lem:corr2}, and \cref{thm:micp}.
\end{proof}
    
\section{{\boldmath \Prob} parameterized by {\boldmath $\np+\nw+\nd$}}
\label{sec:pwd}
In this section, we present the following result.
\begin{theorem}
\label{thm:fpt(pwd)}
  The problem \Prob is fixed-parameter tractable when parameterized by \mbox{$\np+\nw+\nr$} or when parameterized by \mbox{$\np+\nw+\nd$}.
\end{theorem}

We show the second part of \cref{thm:fpt(pwd)}, that is, \Prob is fixed-parameter tractable when parameterized by $\np+\nw+\nd$.
By \cref{obs:releaseduedate}, from this immediately follows that \Prob is also fixed-parameter tractable when parameterized by $\np+\nw+\nr$.
To prove the second part of \cref{thm:fpt(pwd)}, present a reduction from \Prob to MILP.
Given an instance of \Prob, we create a number of instances of MILP where in each of them, the number of integer variables is upper-bounded by a function of $\np$, $\nw$, and~$\nd$.
We solve each instance using \cref{thm:micp} and prove that the \Prob instance admits a schedule where the weighted number of early jobs is at least $W$ if and only if one of the generated instances admits a solution with objective value at least $W$.
Furthermore, we can upper-bound the number of created MILP instances by a function of $\np$, $\nw$, and $\nd$. 

Given an instance of \Prob, we say that two jobs $j$ and $j'$ have the same type if they have the same processing time, the same weight, and the same due date, that is, $p_j=p_{j'}$, $w_j=w_{j'}$, and $d_j=d_{j'}$.
Let $\mathcal T$ denote the set of all types.
Note that we have $|\mathcal{T}|\le \np\cdot\nw\cdot\nd$.
For some $r \in \mathbb{N}_0$ and some $t\in \mathcal{T}$ we denote by $t(r)$ the set of jobs with type $t$ whose release date is at least $r$.
Further, we denote by $p_t$ the processing time of jobs with type $t$.
Let $d_1,\dots, d_{\nd}$ be the sorted sequence of due dates.
We denote with $d_\ell$ the $\ell^{\textnormal{th}}$ due date and we use $d_j$ to denote the due date of job $j$ (same with release dates). 
To keep the notation concise we will use $d_0 = 0$ occasionally.
    
We fix some optimal schedule $\sigma: \{1,\ldots,n\} \to \mathbb{N}$ for the instance so that we may guess some part of it by enumeration.
If for a job $j$ and a due date $d_\ell$ we have $\sigma(j) < d_\ell < \sigma(j) + p_j$, we will say that the job \emph{overlaps} the due date.
Notice that in any schedule, any due date is overlapped by at most one job, but a job may overlap multiple due dates.
    
We now want to enumerate all possible ways due dates can be overlapped by early jobs from some type $t\in\mathcal{T}$ in $\sigma$.
That is, we consider all ways to partition $d_1,\dots, d_{\nd}$ into subsequences of consecutive due dates. There are $2^{d_\#}$ such partitions.
For each such subsequence $S$ we consider all job types that might be scheduled overlapping all due dates in~$S$. 
For the subsequences containing only a single due date, we also consider the case that no job overlaps that due date.
This gives $\np \cdot \nw \cdot \nd + 1$ choices for each subsequence, of which there are~$\nd$.
Thus we end up with at most $2^{\nd} \cdot \nd^{\np \cdot \nw \cdot \nd + 1}$ possible \emph{overlap structures} to consider.
By enumerating them it is now possible to assume that we know which overlap structure is present in $\sigma$. 

We make a small simplification at this stage. 
Suppose we know that some sequence of due dates $d_a,d_{a+1}, \ldots ,d_b$ is overlapped by the same job.
Then $\sigma$ schedules every job with one of these due dates such that it is either scheduled to end before $d_a$, or it is late.
Therefore, we can decrease the due date of such jobs to be $d_a$ without changing the optimality of $\sigma$. 
We may thus assume that every job overlaps at most one due date.
So suppose that for each $d_\ell$ we are given the job type $T(d_\ell) = t$  overlapping that due date, or an assertion that no job overlaps $d_\ell$, represented by $T(d_\ell) = \emptyset$. 
If the due date of jobs of type $T(d_\ell)$ is at most $d_\ell$ we can reject the arrangement immediately, so assume that the due date of jobs of type $T(d_\ell)$ is larger than $d_\ell$.
We can formulate a MILP that computes an optimal schedule under the constraint that this structure of overlaps is respected.
It uses the following variables:
\begin{enumerate}
  \item $x_t^{\ell} \in \mathbb{N}_0$ counts the number of jobs of type $t$ to be scheduled between $d_{\ell-1}$ and~$d_\ell$. \label{var:interiorJobs}
  \item $o_a^\ell, o_b^\ell\in \mathbb{N}_0$ are the portions of the job scheduled overlapping $d_\ell$ that is processed before and after $d_\ell$, respectively. \label{var:overlaps}
  \item $x_j \in [0,1]$ indicates whether job $j$ is scheduled. These variables are fractional, but we will be able to show that they can be rounded.
\end{enumerate}
We could omit generating variables $x_t^\ell$ that are known to schedule jobs late, that is, those where jobs of type $t$ have a due date that is earlier than $d_\ell$. 
For simplicity, we keep these variables, but one can assume them to be set to $0$.
    
The MILP needs the following sets of constraints.
The first constraint sets handle the overlaps around due dates:
\begin{align}
  \forall \ell\in\{1,\ldots,\nd-1\} &: o_a^\ell+ o_b^\ell = p_t, \text{ if } T(d_\ell) = t. \label{con:overlapLength}\\
  \forall \ell\in\{1,\ldots,\nd-1\} &: o_a^\ell+ o_b^\ell = 0, \text{ if } T(d_\ell) = \emptyset. \label{con:overlap0}\\
  \forall \ell\in\{1,\ldots,\nd-1\} &: d_\ell + o_b^\ell \leq d_{\ell+1}. \label{con:overlapFeasR}\\
  \forall \ell\in\{1,\ldots,\nd-1\} &: d_{\ell-1} + o_a^\ell \leq d_{\ell} \enspace .\label{con:overlapFeasL} 
\end{align}
The next set of constraints ensures that we do not try to schedule more jobs of a certain type than there are available:
\begin{equation}
\label{con:jobCount}
  \forall t\in\mathcal{T} : \sum_{j \text{ has type } t} x_j = \sum_{\ell\in\{1,\ldots,\nd\}}x_t^\ell + |\{d_\ell \mid T(d_\ell) = t\}| \enspace .
\end{equation}

The next two sets of constraints, intuitively, ensure that we respect the release dates of the jobs, and that we do not schedule too many jobs between two consecutive due dates.
\begin{align}
  & \forall \ell\in\{1,\ldots,\nd\}, \ell'\in\{1,\ldots,\nr\} \text{ with } r_{\ell'}\le d_\ell \nonumber \\
  &    o_a^\ell + \sum_{t\in\mathcal{T}}p_t\cdot \left(\sum_{j \in t(r_{\ell'})} x_j - \sum_{\ell'' > \ell} x_t^{\ell''} -  |\{d_{\ell''} \mid T(d_{\ell''}) = t , \ell'' \geq \ell\}|\right)\leq d_\ell - r_{\ell'} \enspace . \label{con:feasPerRelease}
\end{align}
\begin{equation}
\label{con:feasTotal}
  \forall \ell\in\{1,\ldots,\nd\} : o_a^\ell + o_b^{\ell-1} + \sum_{t\in\mathcal{T}}p_t\cdot x_t^\ell \leq d_\ell - d_{\ell-1} \enspace .
\end{equation}

The objective function (to be maximized) of the MILP is simply 
\begin{equation}
\label{objective}
  \sum_j w_j\cdot x_j \enspace .
\end{equation}
We observe the following:
\begin{observation}
\label{obs:instvars}
  The number of created MILP instances is in $O(2^{\nd} \cdot \nd^{\np \cdot \nw \cdot \nd + 1})$ and the number of integer variables in each instance is in $O(\np \cdot \nw \cdot \nd^2)$.
\end{observation}

Next, we show the correctness of the reduction.
\begin{lemma}
\label{lem:MILP>OPT}
  If the \Prob instance admits a schedule where the weighted number of early jobs is~$W$, then one of the created MILP instances admits a feasible solution that has objective value at least~$W$.
\end{lemma}
\begin{proof}
  Suppose that we are given a schedule $\sigma$ for the \Prob instance where the weighted number of early jobs is $W$.
  Consider the MILP instance that is generated using the overlap structure of~$\sigma$.
  We create a feasible solution for this MILP instance as follows.
    
  We set $x_j = 1$ if job $j$ is early in $\sigma$, and $0$ otherwise. 
  For every due date $d_\ell$, if there exists a job~$j$ overlapping that due date we set $o^\ell_a = d_\ell - \sigma(j)$ and $o^\ell_b = p_j - o^\ell_a$.
  The variables~$x_t^\ell$ we set to the number of jobs of type $t$ that $\sigma$ schedules such that both their start and completion time are between~$d_{\ell-1}$ and~$d_\ell$.

  If we are using the correct overlap structure for the MILP, we know that this fulfills constraints \eqref{con:overlapLength}, \eqref{con:overlap0}, \eqref{con:overlapFeasR}, and \eqref{con:overlapFeasL}.
  Constraint~\eqref{con:jobCount} is also trivially fulfilled since this is merely counting the number of early jobs in two different ways.
  For constraint~\eqref{con:feasTotal} we observe that the left-hand side counts the total length of (parts of) jobs scheduled by $\sigma$ to run between $d_{\ell-1}$ and $d_\ell$, which must be less than $d_\ell - d_{\ell-1}$ since the jobs do not overlap.
  We can use a similar argument to see that constraint~\eqref{con:feasPerRelease} must be fulfilled.

  Note that $\sum_{t\in\mathcal{T}} p_t \cdot \sum_{j \in t(r_{\ell'})} x_j$ is at most the total length of early jobs scheduled by $\sigma$ after $r_{\ell'}$.
  Meanwhile, $\sum_{t\in\mathcal{T}} p_t \cdot (\sum_{\ell'' > \ell} x_{t}^{\ell''} -  |\{d_{\ell''} \mid T(d_{\ell''}) = t , \ell'' \geq \ell\}|)$ is the total length of early jobs scheduled after $d_\ell$, or overlapping $d_\ell$.
  So we see that the total length of parts of jobs running after $d_\ell$ in $\sigma$ is $\sum_{t\in\mathcal{T}} p_t \cdot (\sum_{\ell'' > \ell} x_{t}^{\ell''} -  |\{d_{\ell''} \mid T(d_{\ell''}) = t , \ell'' \geq \ell\}|) - o^\ell_a$.
  Thus the left-hand side of constraint~\eqref{con:feasPerRelease} is at most the length of parts of jobs scheduled by $\sigma$ to run between $r_{\ell'}$ and $d_\ell$, and is therefore at most $d_\ell - r_{\ell'}$.
\end{proof}

It remains to show that we can construct a schedule from a solution to the MILP.
We will first show some auxiliary result for the case of a single due date.
\begin{lemma}
\label{lem:feasibilitySingleDeadline}
  Consider any instance $I$ of \Prob with a common due date $d$. 
  If for all release dates $r_j$ we have
  \begin{equation*}
    \sum_{j'\in\mathcal I \mid r_{j'} \geq r_j} p_{j'} \leq d-r_j \enspace .
  \end{equation*}
  then we can schedule all jobs early.
\end{lemma}
\begin{proof}
  The statement holds if the instance $I$ contains a single job.
  Otherwise, we apply induction on the number of jobs.
  Let $j^*$ be a job with maximum release date. 
  We schedule that job at time $d-p_{j^*}$.
  This yields a new instance $I'$ of \Prob with one fewer job and common due date $d-p_{j^*}$.
  For any of the release dates $r'_j$ of this new instance, it holds that 
  \begin{equation*}
    \sum_{j'\in I'\mid r_{j'} \geq r'_j} p_{j'} = \sum_{j'\in I' \mid r_{j'} \geq r'_j} p_{j'} - p_{j^*} \leq d-r'_j-p_{j^*} \enspace .
  \end{equation*}
  The lemma statement follows.
\end{proof}
\begin{lemma}
\label{lem:MILP<OPT}
  If one of the created MILP instances admits a solution that has objective value $W$, then the \Prob instance admits a schedule where the weighted number of early jobs is at least $W$.
\end{lemma}
\begin{proof}
  Assume we are given a feasible solution with objective value $W$ to one of the MILP instances.
  We will begin by rounding the $x_j$ such that they are integral. 
  Suppose there is some job $j$ with $x_j \in (0,1)$. 
  Due to constraint~\eqref{con:jobCount} there exists some other job $j'$ with the same type as $j$ and with $x_{j'}\in (0,1)$. 
  Assume, without loss of generality, that $j$ has an earlier release date than $j'$.
  Then we claim that setting $x_j := \min\{1, x_j + x_{j'}\}$ and $x_{j'} := \max\{0, x_{j'}+x_{j}-1\}$ maintains feasibility.
  Note that $x_j + x_{j'}$ does not change, $x_{j'}$ decreases, and in any constraint where $x_j$ occurs, $x_{j'}$ occurs also because $j$ has the earlier release date and they have the same type. 
  Therefore, the left-hand sides of constraint~\eqref{con:jobCount} are unchanged, and those of constraint~\eqref{con:feasPerRelease} do not increase, maintaining feasibility. 
  Since $j$ and $j'$ have the same type, the objective value is unchanged.
  By repeating this rounding operation, we can increase the number of integral $x_j$ until all of them are integral.
  The same argument can also be used to ensure that no job $j$ has $x_j=1$ if there is another job $j'$ with the same type and an earlier release date having $x_{j'} = 0$.

  We now show how to schedule every early job $j$ with $x_j = 1$. Note that this implies that the weighted number of early jobs is at least $W$.
  First, let $J_{t}^1,\dots, J_{t}^{r_\#}$ be the list of jobs with type $t$ sorted by their release dates, and omitting all jobs with $x_j = 0$. 
  It will suffice to prove that we can schedule the last $x_{t}^{\nd}$ jobs from each list into the interval $\mathcal{I} := [d_{\nd-1} + o_r^{\nd-1}, d_{\nd}]$, and that we can then schedule a job of class $S(d_{\nd-1})$ into $[d_{\nd-1} - o_a^{\nd-1}, d_{\nd-1} + o_b^{\nd-1}]$.
        
  Notice that constraint \eqref{con:feasPerRelease} simplifies for $d_{\nd}$ to be $\sum_{t\in\mathcal{T}} \sum_{j \in t(r_{\ell'})} p_t \cdot x_j \leq d_{\nd} - r_{\ell'}$.
  Further, for the purposes of scheduling into $\mathcal{I}$, we may consider the  $x_{t}^{\nd}$ jobs from each type~$t$ with the latest release dates to have release date at least $d_{\nd-1} + o_b^{\nd-1}$, which transforms constraint \eqref{con:feasTotal} into $\sum_{t\in\mathcal{T}} \sum_{j \in t(d_{\nd-1} + o_b^{\nd-1})} p_t\cdot x_j \leq d_{\nd} - d_{\nd-1} - o_b^{\nd-1}$.
  From \cref{lem:feasibilitySingleDeadline} we see that we can indeed schedule all the required jobs into $\mathcal{I}$.

  Remove all the jobs we just scheduled from the instance, and update the MILP solution by setting to $0$ the $x_{t}^{\nd}$, as well as all $x_j$ for the scheduled jobs.
  Now we need to schedule the job~$j^*$ with the latest due date from the jobs of type $T(d_{\nd-1})$ into $[d_{\nd-1} - o_a^{\nd-1}, d_{\nd-1} + o_b^{\nd-1}]$.
  By constraints \eqref{con:overlapLength}-\eqref{con:overlapFeasL} we see that the interval has the correct length for the job, and that the job will not be late. 
  We only need to ensure that it is not scheduled before its release date~$r_{j^*}$.
  We use constraint \eqref{con:feasPerRelease} to observe
  \begin{align*}
    & o_a^{\nd-1} + \sum_{t\in\mathcal{T}}p_t\cdot \left(\sum_{j \in t(r_{j^*})} x_j - \sum_{\ell'' > {\nd-1}} x_{t}^{\ell''} -  |\{d_{\ell''} \mid T(d_{\ell''}) = t , \ell'' \geq {\nd-1}\}|\right) \leq d_{\nd-1} - r_{j^*}\\
    &\implies r_{j^*} - p_{j^*} + \sum_{t\in\mathcal{T}}p_t\cdot \sum_{j \in t(r_{j^*})} x_j \leq d_{{\nd-1}} - o_a^{\nd-1}\\
    &\implies r_{j^*} \leq d_{{\nd-1}} - o_a^{\nd-1}.
  \end{align*}
  We see that $j^*$ can be scheduled as desired.
  Now we update the MILP solution again by setting $d_{\nd-1}$ to $d_{\nd-1} - o_a^{\nd-1}$, as well as $x_{j^*}$, $o_b^{\nd-1}$, and $o_a^{\nd-1}$ to $0$. We also update $T(d_{\nd-1}) := \emptyset$.
  This yields a feasible solution to the MILP of the residual instance where we discard the constraints for the final due date. 
  We can iterate this until all jobs $j$ with $x_j=1$ have been scheduled. We schedule the remaining jobs in an arbitrary but feasible way. Since the objective value of the solution for the MILP is $W$, we have that the weighted number of early jobs is also at least $W$.
\end{proof} 

Now we have all the pieces available that are needed to prove \cref{thm:fpt(pwd)}.

\begin{proof}[Proof of \cref{thm:fpt(pwd)}]
Between \cref{lem:MILP>OPT} and \cref{lem:MILP<OPT} we see that the MILP is equivalent to the original scheduling problem if the correct overlap structure is chosen. 
By \cref{obs:instvars} the number of MILPs we need to solve and the number of integer variables in each MILP depends only on $\np+\nw+\nd$, and thus \cref{thm:fpt(pwd)} follows from \cref{thm:micp} and \cref{obs:releaseduedate}.    
\end{proof}

\section{Unary {\boldmath \Prob} parameterized by {\boldmath $\nr$}}
\label{sec:r&d}

Recall that \Prob is \emph{weakly} $\mathsf{NP}$-hard in the case where there is only one due date and one release date.
Thus even for instances of \Prob with constant $\nr$ or $\nd$, we cannot expect a polynomial-time algorithm solving them in general.
However, we show in the following that such instances can be solved in \emph{pseudo-polynomial time}\footnote{A problem can be solved in \emph{pseudo-polynomial} time if it can be solved in polynomial time when all numeric values are encoded unarily.}, using dynamic programming, thus generalizing the folklore algorithm for {\sc Knapsack}. Formally, we prove the following:

\begin{theorem}\label{thm:XP}
  The problem \Prob is in $\mathsf{XP}$ when parameterized by $\nr$ or $\nd$ if all numbers are encoded in unary.
\end{theorem}

By \cref{obs:releaseduedate}, we can focus on~$\nr$ as a parameter, and the result for $\nd$ as a parameter follows.
We begin by ordering the release dates and denote with $r_\ell$ and $r_k$ the $\ell^{\textnormal{th}}$ and $k^{\textnormal{th}}$ release date, respectively, and we use $r_i$ and $r_j$ to denote the release date of jobs $i$ and $j$, respectively (same with due dates). 
We also use~$\leq_d$ to denote a fixed order of the jobs such that their due dates are non-descending.
We encode this directly into the indexing of the jobs, that is, $i\leq_d j$ if and only if $i \leq j$. 

We can now assume that there exists some optimal schedule such that at every release date $r_\ell$ there is a job scheduled to start exactly at $r_\ell$. 
This can be ensured by enumerating the possible overlap structures of the optimal schedule with respect to the release dates, as in \cref{sec:pwd}.
To do this, we need to guess for each release date $r_\ell$ if there is a job scheduled there and what the completion time of that job is.
We know the completion time to be in $[r_\ell+1,r_\ell+p_{\text{max}}]$, so in reality we will need to solve $(p_{\text{max}})^{\nr}$ dynamic programs and return the best solution found.
This is only a pseudo-polynomial time overhead, so we will suppress it in the following.

Notice that between two consecutive release dates, we can now assume the scheduled early jobs to be ordered according to $\leq_d$, that is, by the earliest due date (all tardy jobs are scheduled later on in a feasible but arbitrary way).
If they are not ordered as such two adjacent out-of-order jobs can be swapped while maintaining feasibility. 

This structural simplification allows us to write a dynamic program with the following recursive definition of the dynamic programming table:
\[
  T[j,t_1,\dots,t_{\nr}] = \max_{\substack{\text{$\ell$ with $r_\ell\geq r_j$}\\  \text{and $r_\ell  + t_\ell\leq d_j$}}}\{ T[j-1,t_1,\dots,t_{\nr}], T[j-1,t_1,,\dots,t_{\ell}-p_j,\dots,t_{\nr}] + w_j\}.
\]

The base cases for the recursion are:
\begin{itemize}
  \item $T[0,0,\dots,0] = 0$,
  \item $T[0,t_1,\dots,t_{\nr}] = -\infty$ if any of the $t_\ell$ are not zero,
  \item $T[j,\cdot,t_\ell,\cdot] = -\infty$ if $t_\ell > r_{\ell+1} - r_\ell$ or $t_\ell<0$ for some $\ell$.
\end{itemize}

The entry $T[j,t_1,\dots,t_{\nr}]$ intuitively represents the most valuable schedule attainable using only the first $j$ jobs according to $\leq_d$, and with a total scheduled job time of $t_\ell$ starting at $r_\ell$. 

We will first notice that the DP-table $T$ has at most $n\cdot (n\cdot p_{\text{max}})^{\nr}$ finite entries. 
Furthermore, each of those entries can be computed in time $O(\nr)$ if we process them in increasing order of the first index.

It remains to show that there exists a schedule attaining the value $\max_{t_1,\dots,t_{\nr}} (T[|J|,t_1,\dots,t_{\nr}])$, as well as that there is some cell $T[|J|,t_1,\dots,t_{\nr}]$ that has value at least as high as that of an optimal schedule.

\begin{lemma}
\label{lem:DPrealizesOPT}
  If the \Prob instance admits a schedule where the weighted number of early jobs is~$W$, then there exists some entry of the dynamic programming table $T$ with value at least $W$.
\end{lemma}
\begin{proof}
  Let $\sigma$ be an optimal schedule for the instance attaining value $W$. 
  We begin by assuming that $\sigma$ schedules the jobs in an order according to $\leq_d$ between two consecutive release dates.
  Now let us denote with $W_i$ the weighted number of early jobs among the first~$i$ jobs according to $\leq_d$, so:
  \begin{equation*}
    W_i = \sum_{j\mid j \leq i, \sigma(j) + p_j \leq d_j} w_j \enspace .
  \end{equation*}
  Further, we denote by $t_{\ell}^i$ the amount of processing time of the first $i$ items scheduled starting at $r_\ell$ up to $r_{\ell+1}$.
    
  We can now show inductively that for every $i$, we can find a cell of the dynamic programming table representing the partial schedule of the first $i$ items, that is, 
  \begin{equation*}
    T[i,t^i_{1}, \dots,t^i_{\nr}] \geq W_i \enspace .
  \end{equation*}
  Note that the base case of the induction with $i = 0$ is true by definition.
    
  So we consider some fixed $i$ and assume $T[k,t^k_{1}, \dots,,t^k_{\nr}] \geq W_k$ for all $k < i$.
  Suppose $\sigma$ job $i$ tardy.
  Then we have $W_{i} = W_{i-1}$ and and $t^i_{\ell} = t^{i-1}_{\ell}$ for all $\ell$, and therefore
  \begin{equation*}
    T[i,t^i_{1}, \dots,t^i_{\nr}] \geq T[i-1,t^i_{1}, \dots,,t^i_{\nr}] \geq W_{i-1} = W_i \enspace .
  \end{equation*}
    
  Otherwise, job $i$ was scheduled early, say after release date $r_k$ but before $r_{k+1}$.
  Then we know that $W_i = W_{i-1} + w_i$, $t^i_k = t^{i-1}_k + p_i$, and $t^i_\ell = t^{i-1}_{\ell}$ for all $\ell \not = k$. 
  Consequently, we get that
  \begin{equation*}
    T[i,t^i_{1},\dots, t^{i-1}_k + p_i ,\dots,t^i_{\nr}] \geq T[i-1,t^{i-1}_{1}, \dots,t^{i-1}_{\nr}] + w_i \geq W_{i-1} + w_i = W_i \enspace .
  \end{equation*}

  This concludes the induction. 
  We finally note that because $\sigma$ is a schedule the cell $T[n,t^n_1,\dots,t^n_{\nr}]$ has a finite value, and from the inductive argument that value is at least $W$.
\end{proof}

\begin{lemma}
\label{lem:DPisRealizable}
  Let $T[i,t_1,\dots,t_{\nr}]$ be a cell of the dynamic program with finite value~$W$. 
  Then there is a schedule where at most the first $i$ jobs are early, that attains a weighted number of early jobs $W$, and which schedules a total processing time of $t_\ell$ immediately after release date $r_\ell$.
\end{lemma}
\begin{proof}
  Again we prove the lemma by induction, noting that the statement holds for $T[0,0,\dots,0]$ with a schedule that schedules all jobs to be tardy.
  So fix some $i$ and assume the result for all $k < i$.

  If $T[i, t_1, \dots, t_{\nr}] = T[i-1, t_1, \dots, t_{\nr}]$ we use the schedule guaranteed by induction for $T[i-1, t_1, \dots, t_{\nr}]$ and in addition we schedule the $i$-th job to be tardy.
  Otherwise, there is some $k$ such that
  $T[i, t_1, \dots, t_k, \dots, t_{\nr}] = T[i-1, t_1, \dots, t_k-p_i, \dots , t_{\nr}] + w_i$.
  So we take the schedule guaranteed for $T[i-1, t_1, \dots, t_k-p_i, \dots , t_{\nr}]$ and additionally schedule job $i$ at time $r_k + t_k-p_i$.
  From the definition of $T[\cdot]$ we have that $r_i \leq r_k$, so the job may be scheduled at this time. 
  We also get $r_k+t_k \leq d_i$, so the job is early and so the value of this schedule will be $T[i-1, t_1, \dots, t_k-p_i, \dots , t_{\nr}] + w_i$.
  Finally, observe that we have  $r_k + t_k \leq r^{k+1}$, so the inductive hypothesis is maintained and so the $i$-th job does not conflict with any other jobs already scheduled.
  Thus we have constructed a schedule satisfying the induction hypothesis.
\end{proof}

We may now conclude the following:

\begin{corollary}\label{cor:corr}
  Let $W = \max_{t_1,\dots,t_{\nr}}\{ T[n,t_1,\dots,t_{\nr}]\}$.
  Then $W$ is the maximum weighted number of jobs that can be scheduled early in the \Prob instance.
\end{corollary}
\begin{proof}
  By \cref{lem:DPrealizesOPT} we have that the maximum weighted number of jobs that can be scheduled early in the \Prob instance is ar least $W$.
  However, let $T[n, t_1,\dots, t_{\nr}]$ be an entry with value~$W$.
  Then \cref{lem:DPisRealizable} guarantees that there exists a schedule for the \Prob instance such that the weighted number of early jobs is $W$.
\end{proof}

Now we have all the pieces to prove a running time upper bounds for our dynamic-programming algorithm for \Prob.

\begin{lemma}\label{lem:correctness}
  The problem \Prob can be solved in time $O(p_{\text{max}}^{\nr} \cdot n\cdot (n\cdot p_{max})^{\nr}\cdot \nr)$. 
\end{lemma}
\begin{proof}
  We first need to guess the correct overlap structure between jobs and release dates of which there are at most $p_{\max}^{\nr}$ many.
  For each of these overlap structures, we then compute the dynamic programming in time $O(n\cdot (n\cdot p_{\max})^{\nr}\cdot \nr)$. 
  We finally return the best objective value found in any of the dynamic programming tables. By \cref{cor:corr}, this is the maximum weighted number of jobs that can be scheduled early in the \Prob instance.
\end{proof}

Now we have all the pieces to prove \cref{thm:XP}. 

\begin{proof}[Proof of \cref{thm:XP}]
    \cref{thm:XP} follows directly from \cref{lem:correctness} and \cref{obs:releaseduedate}.
\end{proof}

Notice that the proof of \cref{lem:DPisRealizable} also gives a backtracking procedure which allows us to compute an optimal schedule in time $O(n\cdot \nr)$ for a given filled in dynamic programming table with the maximum entry already computed.

\section{Conclusion}
In this work, we give a comprehensive overview of the parameterized complexity of \Prob for the parameters number $\np$ of processing times, number $\nw$ of weights, number $\nr$ of release dates, and number $\nd$ of due dates. We leave several questions for future research, in particular:
\begin{itemize}
    \item Is \Prob in $\mathsf{XP}$ when parameterized by $\np$?
    \item Is \Prob fixed-parameter tractable when parameterized by $\np + \nw$?
    \item Is \Prob fixed-parameter tractable when parameterized by $\np$ and all numbers are encoded unarily?
\end{itemize}
We remark that a technical report~\cite{ElffersW14} claims (strong) $\mathsf{NP}$-hardness even for $\np=2$ and $\nw = 1$.
If that result holds, then it would also settle the parameterized complexity for the parameter combination $\np+\nw$, and for $\np$ even when all processing times and weights are encoded in unary.

\paragraph{Acknowledgements.} We wish to thank Danny Hermelin and Dvir Shabtay for fruitful discussions that led to some of the results of this work.

\bibliography{bibliography}

\end{document}